\documentclass[a4paper,11pt]{article}

\usepackage[utf8]{inputenc}
\usepackage{amsmath}
\usepackage{amssymb}
\usepackage{amsthm}
\usepackage{graphicx}
\usepackage{float}

\newtheorem{prop}{Proposition}

\DeclareMathOperator{\Var}{{Var}}

\title{A statistical approach for robust tolerance design}
\author{Ambre~Diet$^{1,2}$, Nicolas~Couellan$^{2,3}$, Xavier~Gendre$^{2,4}$, Julien~Martin$^1$}
\date{}

\newcommand\blfootnote[1]{%
  \begingroup
  \renewcommand\thefootnote{}\footnote{#1}%
  \addtocounter{footnote}{-1}%
  \endgroup
}

\begin{document}
\maketitle

\blfootnote{$^1$ Tolerancing department, Airbus Operations S.A.S, 316 route de Bayonne, 31060 Toulouse, France.}
\blfootnote{$^2$ Institut de Mathématiques de Toulouse UMR 5219, Université de Toulouse, 31062 Toulouse, France.}
\blfootnote{$^3$ ENAC, Université de Toulouse, 7 avenue Édouard Belin, 31400 Toulouse, France.}
\blfootnote{$^4$ ISAE-SUPAERO, Université de Toulouse, 10 Avenue Édouard Belin, 31055 Toulouse, France.}

\begin{abstract}
Within an industrial manufacturing process, tolerancing is a key player. The dimensions uncertainties management starts during the design phase, with an assessment on variability of parts not yet produced. For one assembly step, we can gain knowledge from the tolerance range required for the parts involved. In order to assess output uncertainty of this assembly in a reliable way, this paper presents an approach based on the deviation of the sum of uniform distributions. As traditional approaches based on Hoeffding inequalities do not give accurate results when the deviation considered is small, we propose an improved upper bound. We then discuss how the stack chain geometry impacts the bound definition. Finally, we show an application of the proposed approach in tolerance design of an aircraft sub-assembly. The main interest of the technique compared to existing methodologies is the management of the confidence level and the emphasis of the explicit role of the balance within the stack chain.  
\end{abstract}

\noindent\textit{Keywords: Design, Deviation, Manufacturing, Quality, Sum of uniform distributions, Tolerance}

\paragraph{Note to Practitioners} This paper was motivated by the problem of balance between conservative approach and manufacturing cost when tolerance are allocated during design phase. Without information available about parts dimensions distributions, tolerance intervals must be restrictive enough to ensure the feasibility of assembly but also compliant with industrial capabilities and cost optimization. The main methods currently applied are the worst case approach and the statistical approach based on Gaussian assumption. We mathematically define a methodology in between to allocate tolerance, with the advantage of allowing the management of confidence level. As we assume uniform distributions on the parts dimensions, our proposal is also robust against non-normal distributions which can happen when the industrial process is not perfectly followed or when measurement are not available. Finally, we highlight the role of balance between stack chain contributors and quantify its link with the restrictiveness of tolerance interval considering the statistical result as a reference. In future research, we propose to study how to improve the method with different probability laws such as bi-modal or truncated distributions and how we can generalize the approach when there are multiple top level requirements for a part dimension.

\section{Introduction}

The management of dimensions uncertainties is a key player in the manufacturing process of various industrial sectors such as transportation (automotive, aeronautics, \dots) or household appliances industry. 

Dimensions may have some deviation from the designed value without significant impact on the quality and functional requirements of the final product. Tolerance intervals are defined according to engineering knowledge and scientific analysis in order to determine these acceptable variations.  A deviation out of the determined tolerance bounds is considered non-compliant and imply an action such as an investigation or a modification in the process or the design. 

The perfect balance between functional requirements and process capability has to be found so that the specified tolerance interval is the most accurate possible. If the tolerance is too tight, the process might not have the capability to manufacture it and either there will be many rejected items either some costly improvement will be needed to produce compliant items. Otherwise, a too wide tolerance will lead to non-conformity with functional requirements of the final product and may lower the aircraft performance. 
As there are often several steps in a manufacturing process, the propagation of uncertainty has also to be taken into account to specify the tolerance interval of following assembly steps.

 In this paper, we focus on the tolerances allocation during the design phase in which tolerancing activity do not only aim at anticipating the margins of uncertainty but also help in predicting their effects on the various assembly steps. These involve different physical characteristics of parts, such as part length, hole position, pin, \dots, called features. All tolerancing issues and notations are detailed in the engineering drawing and related documentation pratices \cite{ASME Y14} and \cite{ISO 1101}. \\

In our case, there are no available dimensions measurements because we focus on tolerance allocation in the design phase of a product prior production. Considering one specific assembly stage, one of the main concern is to assess the variability of an output feature of the assembly knowing the tolerance range of the input features. In Figure~\ref{ex1}, input features are the lengths of the different parts and the output feature is the gap of interest in this assembly.

\begin{figure}[H]
\centering
\includegraphics[scale=0.57]{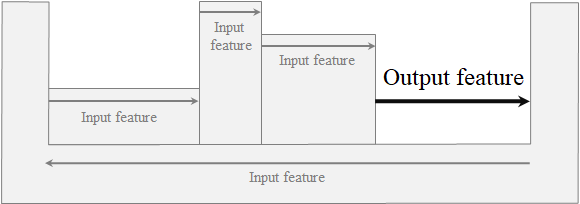}
\caption{An assembly example : inputs and output features identification.}
\label{ex1}
\end{figure}

 In the design phase, all tolerance intervals are assumed to be centered around the nominal dimension. To determine the variation around this nominal dimension, there are two main methods detailed in \cite{CHASE 88} : Worst Case and statistical approaches. 

Worst Case approach is to consider all assembly parts delivered at their worst acceptable value (assembly output tolerance equals the sum of the input tolerances). Statistical approach, also called RSS approach (square Root of the Sum of Squares), gives a result assuming all input features are normally distributed (assembly output tolerance equals the square root of the sum of squares of input features tolerances). Statistical result gives a much tighter tolerance range result than the worst case approach, but it does not hedge against the case where input are not reasonably close to their nominal value. To find a balance between this two approaches, Bender \cite{BENDER 62} proposed to multiply the statistical result by an empiric coefficient of 1.5 to obtain an inflated statistical result which is supposed to give a result tighter than the worst case but more conservative than RSS. However, this technique does not apply if number of assembly inputs is low as it gives a wider result than the worst case approach. Several other statistical methodologies have been studied to obtain the best trade-off between worst case and statistical approaches. For instance, Skowronski and Turner \cite{SKO TURNER 97} proposed a method relying on Monte Carlo techniques. Choi \textit{et al.} \cite{CHOI 00} studied an approach based on Taguchi's method requiring the definition a quadratic loss function. The tolerance allocation problem is formulated as a minimization of the sum of machining cost and quality loss. Manufacturing cost considerations for tolerance allocation is beyond the scope of this article. Pillet \textit{et al.} \cite{PILLET 05} proposed to consider weighted inertial tolerancing. Inertial tolerancing works with mean square deviation (inertia) of the output feature as limit instead of considering a tolerance interval. Then, they applied a weighting system based on the number of assembly inputs to obtain a reasonable tolerance result. An other approach has been studied in \cite{LEBLOND PILLET 18} by taking an interest in the meaning of the conformity. Instead of limiting the assembly output variability, they propose a formal definition of statistical conformity that does not apply individually to a part but to a part population.

Note that tolerance intervals are highly related to the assemblies processes capabilities. Even if suppliers process capability indicators should be monitored as detailed in \cite{DRAKE 99}, the normality of features distributions can not always be verified.

One of the objective of the tolerancing is to assess the same confidence in a tolerance interval whatever the distribution of inputs are, as long as these inputs are delivered within the claimed tolerance range. Indeed, suppliers of parts receive a nominal value and two dimension limits. They are also required to follow a target distribution, however checking this compliance is difficult in practice. At the design stage, it is impossible to characterize the features distributions from measurement data. Uniform distribution is a better option to hedge against less favorable distributions of suppliers values.

We propose a mathematical tool to define an accurate assembly output tolerance range considering uniform input features. 

The paper is organized as follows: statistical framework is introduced in Section~\ref{stat}, main results are presented in Section~\ref{main}: First part is devoted to traditional approach on deviations and following parts detail improvement on the upper bound accuracy and balance term introduction. In Section~\ref{appli}, we carry out a simulation study in order to represent and compare our results. Finally, an example on airframe assembly with real inputs data is performed.  

\section{Statistical framework} \label{stat}

Consider a set of input features $X_1, \ldots, X_n \in \mathbb{R}$ and an output feature $Y \in \mathbb{R}$. All input features are assumed to be independent random variables for the reason that assembly parts are supposed to be separately produced.  We are interested here in the variability of the output $Y$ and especially in a way to define a tolerance range for this feature. \\

Each input feature is assumed to be centered around a nominal dimension and has its own variability characterized by its tolerance range $[-v_i,v_i], \forall i \in \{1,\ldots,n\}$ where $v_1, \ldots, v_n > 0$ are the tolerance bounds. This variability reflects the uncertainty linked to the process (temperature, control plan, ground motion, delivery types, \dots). \\

In order to discuss about the feature $Y$, assembly step must be modeled to represent the link between inputs and output of the assembly. A common approach in tolerancing is the linear coefficient model. If the variations are supposed to be small around the nominal dimension, the linear approach is appropriate.

Based on the knowledge of inputs tolerances and influence coefficients on the output, output result is seen as a linear combination of all inputs weighted by known influence coefficients (previously determined with a 3D CAD tool and only linked to the assembly geometry). For ease of notations, we will directly treat the weighted features, meaning that input features $X_1, \ldots, X_n$ are already multiplied by their respective influence coefficients. In this formalism, 
\begin{equation*}
Y =\sum_{i=1}^n X_i.
\end{equation*}

 For a given confidence level $\rho$, the aim is to determine the associated tolerance interval $[-t,t]$ for the output feature $Y$, verifying 
\begin{equation*}
\mathbb{P}(\left| Y \right| \geqslant t) \leqslant \rho.
\end{equation*}

The tolerance interval is determined based on the distribution of the output feature which depends on input features distributions.\\

 A popular practice is to consider all features as Gaussian which leads to a Gaussian output feature. By applying the commonly used $6\sigma$ methodology, the confidence level is $\rho= 0.0027$. In Gaussian framework, the result is $t= 3\sigma_Y$ with $\sigma_Y$ the standard deviation of the feature $Y$. \\

Within the Gaussian framework, the $6\sigma$ methodology gives the standard deviation of each input feature :  $v_i/3, \forall i \in {1,\ldots,n}$. As input features are assumed independent, the standard deviation of the output feature in the Gaussian case is $\frac{1}{3} \sqrt{\sum_{i=1}^n v_i^2}$. Again, the $6\sigma$ methodology leads to the interval $[-T_{RSS}, T_{RSS}]$ for the output feature tolerance, where $T_{RSS} = \sqrt{\sum_{i=1}^n v_i^2}$. This tolerance interval is commonly called the statistical result or RSS (Root Sum of Squared) result by the tolerancing community. However, as we consider tolerance allocation in the design phase, Gaussian assumption can not be verified from measurement data on features. Only tolerances bounds of input features $v_1, \ldots, v_n$ are available.

We choose to consider input features as uniform random variables, since it is the least informative available distribution given our knowledge about inputs. The purpose is to characterize the deviation of the sum of uniform independent random variables.
Killmann and Von Collani \cite{KILLMANN 01} studied the distribution of the sum of uniform features. Their idea was to explicitly calculate density of the sum but such a closed form is numerically intractable and therefore not suited to our context. 

Note that our objective is to focus on the quantile of the distribution of $Y$ ensuring a given probability $\rho$ to be out of tolerance. This probability value is fixed in our framework and we do not plan to address the tolerancing problem uniformly according to $\rho$. This is the point of view of the field of \textit{optimal transport} as developed in \cite{VILLANI 03} but controlling distribution tails leads to poor results in practice for reasonable values of $\rho$.

 In the uniform case, input features standard deviations are now $v_i/\sqrt{3}, \forall i \in {1,\ldots,n},$ and the standard deviation of the output feature is $\frac{1}{\sqrt{3}}\sqrt{\sum_{i=1}^n v_i^2}$. In this case, the $6\sigma$ methodology is applied with standard deviations of uniform distributions and the output feature tolerance interval would be $[-\sqrt{3} \times T_{RSS},\sqrt{3} \times T_{RSS}]$.

The coefficient $\sqrt{3}$  is an accurate coverage factor on the statistical result if the $v_1, \dots, v_n$ are all equal but it does not address the case where they are unbalanced. Yet, if one of the feature predominates over others, for a same confidence level the output tolerance interval should be tighter, as shown in the Figure~\ref{balance}. 

\begin{figure}[H]
\centering
\includegraphics[scale=0.39]{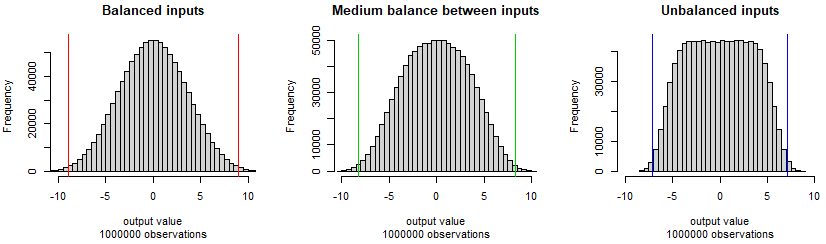}
\caption{Output feature distribution for several balance ratio of inputs.}
\label{balance}
\end{figure}

The aim of our approach is to introduce a shape coefficient in order to correct the RSS interval result assuming the input distributions are uniforms. This coefficient will depend on how unbalanced input features are and also on the selected confidence level $\rho$. The value $\sqrt{3}$ for this shape coefficient means that $v_1 = v_2 = \dots = v_n$. The more input features are unbalanced, the lower the form coefficient value is.

Next, we will focus on the role of this coefficient and its impact on the probability $\rho$.

\section{Main results} \label{main}

If we consider Gaussian independent input features with a tolerance interval $[-v_i,v_i]$, $\forall i \in \{1, \ldots,n\}, $ the associated standard deviation from the $6\sigma$ methodology is $v_i/3$, $\forall i \in \{1, \ldots,n\}$. If the features are denoted $N_i$ and $N_i \sim \mathcal{N}(0, v_i/3)$, then $\forall i \in {1, \ldots,n}$, the standard Gaussian deviation inequality gives

\begin{equation*}
\mathbb{P}\left( \left| \sum_{i=1}^n N_i  \right| \geqslant t \right) \leqslant 2 \exp\left(- \frac{t^2}{2 \sum_{i=1}^n \left( \frac{ v_i}{3}\right)^2} \right)
\end{equation*}
and then
\begin{equation*}
\mathbb{P}\left( \left| \sum_{i=1}^n N_i  \right| \geqslant \frac{1}{3} \sqrt{2 log \left( \frac{2}{\rho} \right) \sum_{i=1}^n v_i^2 } \right) \leqslant \rho
\end{equation*}
that is equivalent to
\begin{equation*}
\mathbb{P}\left( \left| \sum_{i=1}^n N_i  \right| \geqslant l_{\rho} \times T_{RSS}  \right) \leqslant \rho
\end{equation*}
with
\begin{equation*}
l_{\rho} = \frac{1}{3} \sqrt{ 2  log \left( \frac{2}{\rho} \right) }.
\end{equation*}

For fixed $\rho$ and independent uniform input features $U_i \sim \mathcal{U}([-v_i,v_i])$, $\forall i \in {1, \ldots,n} $, our aim is to determine $f$ such that 

\begin{equation} \label{eq:unif}
\mathbb{P}\left( \left| \sum_{i=1}^n U_i  \right| \geqslant f   \times l_{\rho} \times T_{RSS}  \right) \leqslant \rho.
\end{equation}

\subsection{Hoeffding approach for the deviation of a sum of bounded random variables} 

Traditional approaches based on deviations are related to the Hoeffding inequality which provides an upper bound on the probability that the sum of bounded independent random variables deviates more than a certain amount. 

As detailed in \cite{HOEFFDING 63} and \cite{BOUCHERON 13}, this inequality applied to the sum of uniform independent random variables $Y =  \sum_{i=1}^n X_i $  gives a non asymptotic upper bound for the probability of deviation. This result is summarized in the Proposition~\ref{prop1}.

\begin{prop}
\label{prop1}
Let $v_1, \dots, v_n > 0$,
if $X_1, \dots, X_n$ are independent random variables such that
\begin{equation*}
\forall i \in \{1, \dots, n\},\ |X_i| \leqslant v_i, \ a.s.
\end{equation*}
then,
\begin{equation*}
\forall t > 0,
\mathbb{P}\left( \left| \sum_{i=1}^n X_i \right| \geqslant t \right) \leqslant 2 \exp\left(- \frac{t^2}{2\sum_{i=1}^n v_i^2} \right).
\end{equation*}
\end{prop}

\begin{proof}
See Section 2.6 in \cite{BOUCHERON 13}.
\end{proof}

Setting $t=\sqrt{2\log \left(\frac{2}{\rho}\right)\sum_{i=1}^n v_i^2 }$ leads to
\begin{equation*}
\mathbb{P}\left( \sum_{i=1}^n X_i \geqslant \sqrt{2\log \left(\frac{2}{\rho}\right) \sum_{i=1}^n v_i^2 } \right) \leqslant \rho,
\end{equation*}
and taking $f=3$ in \eqref{eq:unif}, we have 
\begin{equation*}
\mathbb{P}\left( \sum_{i=1}^n X_i \geqslant 3 \times l_{\rho} \times T_{RSS} \right) \leqslant \rho.
\end{equation*}

Hoeffding approach only takes into account the fact that random variables are bounded. However, in our case we also have the information that features are uniform random variables. We will use this information to find a tighter upper bound for the deviation of a  sum of uniform random variables. As a result, we will obtain a lower value for the coefficient $f$ we are looking for.

\subsection{Chernov approach to improve the bound for a sum of uniform random variables}

The aim is to find a more accurate upper bound than the previous one coming from the Hoeffding approach. We state and prove the following proposition:

\begin{prop}
\label{prop2}
Let $v_1, \dots, v_n > 0$,
if $X_1, \dots, X_n$ are independent random variables such that
\begin{equation*}
\forall i \in \{1, \dots, n\},\ X_i \sim \mathcal{U}([-v_i,v_i]),
\end{equation*}
then,
\begin{equation*}
\forall t > 0,
\mathbb{P}\left( \left| \sum_{i=1}^n X_i \right| \geqslant t \right) \leqslant 2 \inf_{\lambda>0} \left\{ \exp \left( \phi(\lambda,t) \right) \right\}
\end{equation*}

where the function $\phi$ is defined for any $\lambda, t> 0$ by 

\begin{equation}
\label{eq:phi}
\phi(\lambda,t)  =  \sum_{i=1}^n \log\left( \frac{e^{\lambda v_i}-e^{-\lambda v_i}}{2 \lambda v_i} \right)  - \lambda t.
\end{equation}

\end{prop}

\begin{proof}
Let $\lambda>0$, applying Markov inequality to the positive random variable $\exp \left( \lambda \sum_{i=1}^n X_i  \right) $ gives the following upper bound on the probability that the sum of uniform independent random variables deviates more than $t>0$

\begin{equation*}
\mathbb{P}\left(\sum_{i=1}^n X_i  \geqslant t \right) \leqslant  \frac{\mathbb{E}[e^{ \lambda \sum_{i=1}^n X_i}]}{e^{\lambda t}}.
\end{equation*}

Using the symmetry of the uniform distribution, we have for any $t > 0$

\begin{equation*}
\mathbb{P}\left( \left|\sum_{i=1}^n X_i \right| \geqslant t \right) \leqslant 2 \exp \left( \phi(\lambda,t) \right).
\end{equation*}

The upper bound is valid for any value of $\lambda>0$ and the announced result follows by taking the infimum according to $\lambda$.

\end{proof}

The optimization of $\phi$ function with respect to $\lambda$ is highly related to the input features balance. Next, we detail how to characterize this dependency.

\subsection{Dependency on the features balance}

The aim of this section is to provide details on how to determine the value of $\lambda$ which is obtained from a function minimization in the upper bound previously presented.

A concept of balance between input features is introduced. This balance represents the discrepancy between the uniform distributions  parameters : if all uniform random variables have the same parameters, it means a perfect balance between tolerance bounds. Otherwise, one of the random variables within the sum may have a much larger support set than others and it leads to imbalance between tolerance bounds.

We take the upper bound result from Proposition~\eqref{prop2}. The idea is to bound from above this result by introducing a specific term that identifies the influence of the balance within $v_1, \dots, v_n$. 

We focus on the sum of logarithms in the function $\phi$ given in \eqref{eq:phi} that can be rewritten as
\begin{align}
 \sum_{i=1}^n \log\left( \frac{e^{\lambda v_i}-e^{-\lambda v_i}}{2 \lambda v_i} \right)  &=    \lambda  \sum_{i=1}^n v_i  + \sum_{i=1}^n \log\left( \frac{1-e^{-2\lambda v_i}}{2 \lambda v_i} \right)\nonumber \\
&=   n\log\left( \frac{1-e^{-2\lambda \bar{v}}}{2 \lambda \bar{v}}  \right) + S_{\lambda}
%\label{eq3}
\label{S_intro}
\end{align}
with $S_{\lambda}$ defined as follows
\begin{equation*}
S_{\lambda} =\sum_{i=1}^n \left( \log \left( \frac{1-e^{-2\lambda v_i}}{2 \lambda v_i} \right) -  \log \left( \frac{1-e^{-2\lambda \bar{v}}}{2 \lambda \bar{v}} \right)  \right).
\end{equation*}

The term $S_{\lambda}$ quantifies the imbalance between uniform distributions parameters. In the next two propositions, we propose results about the upper bound on the probability that the sum of uniform independent random variables deviates from its expected value.

\begin{prop}
\label{prop_lipschitz}
Let $v_1, \dots, v_n > 0$, we introduce 
\begin{equation*}
\bar{v} = \frac{1}{n} \sum_{i=1}^n v_i .
\end{equation*}
 If $X_1, \dots, X_n$ are independent random variables such that $\forall i \in \{1, \dots, n\}$,$\ X_i \sim \mathcal{U}([-v_i,v_i])$ then,
\begin{equation*}
\forall t > 0,\ \mathbb{P}\left( \left| \sum_{i=1}^n X_i \right| \geqslant t \right)  \leqslant   2 \exp \left( \psi(\lambda_0,t) \right).
\end{equation*}
where for any $\lambda, t > 0$
\begin{equation*}
\psi(\lambda,t) =
-\lambda t 
+ \lambda n \bar{v}  
+ n \log \left( \frac{1- e^{-2\lambda \bar{v}}}{2 \lambda \bar{v}} \right) 
+  \lambda  \sum_{i=1}^n \left|  v_i - \bar{v}  \right|
\end{equation*}
and $\lambda_0$ is such that
\begin{equation*}
\frac{ \partial   \psi(\lambda_0,t)}{ \partial \lambda} = 0
\end{equation*}
\end{prop}

For a set of tolerance bounds $v_1, \dots, v_n > 0$ and a fixed probability $\rho$, $t$ verifies $  \psi(\lambda_0,t) =\rho$ . The value of interest $t$ is obtained by inversion with respect to $t$ of the function $\psi$. With this expression, the balance within $v_1, \dots, v_n$ appears via $ \sum_{i=1}^n \left|  v_i - \bar{v}  \right| $. Indeed, this term is large for unbalanced values  $v_1, \dots, v_n$ and small otherwise. Next, we prove Propostion~\ref{prop_lipschitz}.

\begin{proof}
Let us define the function $h$ given by

\begin{equation*}
\forall x > 0, \quad 
h(x) = \log \left( \frac{1-e^{-x}}{x} \right).
\end{equation*}

This function is $\frac{1}{2}$-Lipschitz continuous (proof is postponed in the appendix) and therefore we can write

\begin{equation}
\forall x, y > 0, \quad
\left| h(x) - h(y) \right| \leqslant \frac{1}{2} \left| x - y \right|.
\label{lip}
\end{equation}

We apply this inequality for $x=2 \lambda v_i$ $\forall i \in \{1, \dots, n\}$ and for $y=2 \lambda \bar{v}$ and sum the terms to obtain

\begin{equation*}
S_{\lambda}   \leqslant   \lambda  \sum_{i=1}^n \left|  v_i - \bar{v}  \right|.
\end{equation*}
The announced result follows from this upper bound on $S_{\lambda}$ in the equation \eqref{S_intro}.
\end{proof}

In the previous proposition, the balance ratio of the $v_i$ was quantified through the absolute values $|v_i - \bar{v}|$. It is natural to consider also the variance to this end and this is the purpose of the next proposition.

\begin{prop}
\label{prop_var}
Let $v_1, \dots, v_n > 0$, we introduce 
\begin{equation*}
\bar{v} = \frac{1}{n} \sum_{i=1}^n v_i \quad \text{and} \quad \Var(v)=\frac{1}{n}  \sum_{i=1}^n \left( v_i - \bar{v} \right) ^2 .
\end{equation*}
 If $X_1, \dots, X_n$ are independent random variables such that $\forall i \in \{1, \dots, n\}$,$\ X_i \sim \mathcal{U}([-v_i,v_i])$ then,
\begin{equation*}
\forall t > 0,\ \mathbb{P}\left( \left| \sum_{i=1}^n X_i \right| \geqslant t \right)  \leqslant   2 \exp \left(\tilde{\psi}(\lambda_0,t) \right).
\end{equation*}
where for any $\lambda, t > 0$
\begin{equation*}
\tilde{\psi}(\lambda,t) =
-\lambda t 
+ \lambda n \bar{v}  
+ n \log \left( \frac{1- e^{-2\lambda \bar{v}}}{2 \lambda \bar{v}} \right) 
+ \frac{n \lambda^2  \Var(v) }{2} 
\end{equation*}
and $\lambda_0$ is such that
\begin{equation*}
\frac{ \partial   \tilde{\psi}(\lambda_0,t)}{ \partial \lambda} = 0
\end{equation*}
\end{prop}

As for Proposition~\ref{prop_lipschitz}, we have a set of tolerance bounds $v_1, \dots, v_n > 0$ and a fixed probability $\rho$ and value $t$ is obtained by inversion with respect to $t$ of the function $\tilde{\psi}$. The proof of Proposition~\ref{prop_var} is as follows.

\begin{proof}
The Lipschitz continuity of $h$ ensures the following inequality (see for example Lemma 1.2.3 in \cite{NESTEROV 13} for a proof of this result):
\begin{equation*}
\forall x, y > 0, \quad
\left| h(x) - h(y) \right| \leqslant \frac{L}{2} \left \Vert x - y \right \Vert^2.
\end{equation*}
This result applied to $x=2 \lambda v_i$ $\forall i \in \{1, \dots, n\}$ and for $y=2 \lambda \bar{v}$  gives 
\begin{equation*}
\forall \lambda, v_1, \dots, v_n > 0, \quad
\left| h(2 \lambda v_i) - h(2 \lambda \bar{v}) \right| \leqslant \frac{\lambda^2}{2}  \Vert  v_i -  \bar{v}  \Vert ^2
\end{equation*}
and finally, since $S_{\lambda}=\sum_{i=1}^n \left( h(2 \lambda v_i) - h(2 \lambda \bar{v}) \right)$, we obtain
\begin{equation*}
 S_{\lambda}   \leqslant   \frac{n \lambda^2  \Var(v) }{2}.
\end{equation*}
The announced result follows from this upper bound on $S_{\lambda}$ in the equation \eqref{S_intro}.
\end{proof}

\section{Applications} \label{appli}

The first part of this section will describe how our upper bound behaves on different stack chains obtained from simulations. The second part will focus on a practical study on an industrial example of tolerance definition within an aircraft assembly.

\subsection{Simulations}

\subsubsection{Tolerance design on an assembly example}
First step is to simulate stack chains. We randomly generate stack chains with a number of inputs $n=5$. Their tolerance intervals values are also randomly generated between 1 and 5. We would like to assess on the output tolerance variability via tolerance intervals to be defined. We generate a stack chain with tolerance inputs intervals bounds and traditional output results as following:

\begin{table}[H]
\caption{Example of a stack chain characterization.}\label{Example of a Stack chain characterization}
\centering
\begin{tabular}{c|c|c|c|c||c|c}
%\hline
$X_1$  &     $X_2$     &      $X_3$        &          $X_4$   &   $X_5$   & RSS & WC \\ \hline 
$v_1=5$                      & $v_2=4$                      & $v_3=3$                      & $v_4=2$                      & $v_5=1$ & 7.4 & 15 \rule[-7pt]{0pt}{16pt} \\ %\hline                      
\end{tabular}
\end{table} 

The first approach is based on Monte Carlo methods. We generate $N=10^5$ observations from $n=5$ uniform distributions and we sum it. The probability to be out of a given tolerance interval can therefore be asymptotically estimated and considered as a near theoretical result. The two following methods give an output interval bound according to the two approaches we proposed in this paper. The upper bound we provide depends on the selected confidence level $\rho$. This level is the probability for the output feature to be out of the output tolerance interval we design. The higher the confidence level, the wider the tolerance interval is. Indeed, if we allow more values to be out of tolerance, the output tolerance interval should be broader. Figure~\ref{ex_tol_def} illustrates the results obtained from Monte Carlo draws and from the methods detailed in this paper. Among the three methods, Figure~\ref{ex_tol_def} shows that the Lipschitz and Quadratic approaches give a looser upper bound than the Chernov method.

\begin{figure}[H]
\centering
\includegraphics[scale=0.63]{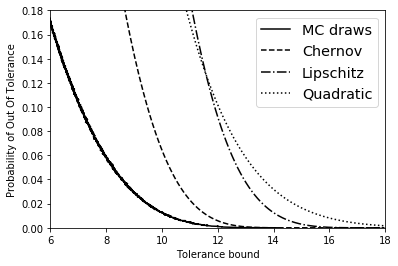}
\caption{Example of the behavior of the probability to be out of tolerance with respect to the value of the output tolerance bound for the discussed approaches.}
\label{ex_tol_def}
\end{figure} 

 The benefits of the methods proposed in this article is that they do not require Monte Carlo draws, nor asymptotic estimation of the probability to be out of tolerance. Indeed, the bounds we provide offer theoretical non asymptotic guarantees and eliminate any risk of rare events that Monte Carlo methods would not generate. Moreover, for large assemblies, the number of Monte Carlo draws needed to obtain a sufficiently sharp result would grow with the number of input features in the assembly. The formula we discuss are closed and directly usable in practice and cheaper to compute than Monte Carlo simulations.

\subsubsection{Influence of the assembly geometry}

In order to represent the balance within a stack chain, previous sections introduced the following term
\begin{equation*}
S_{\lambda} =\sum_{i=1}^n \left( \log \left( \frac{1-e^{-2\lambda v_i}}{2 \lambda v_i} \right) -  \log \left( \frac{1-e^{-2\lambda \bar{v}}}{2 \lambda \bar{v}} \right)  \right).
\end{equation*}
In particular, taking a parameter $\lambda = 1$ leads to
\begin{equation*}
S_1 =\sum_{i=1}^n \left( \log \left( \frac{1-e^{-2 v_i}}{2 v_i} \right) -  \log \left( \frac{1-e^{-2 \bar{v}}}{2 \bar{v}} \right)  \right).
\end{equation*}
This quantity can be used as an indicator of the balance of the stack chain. Indeed, the more balanced the stack chain is, the lower the value is and vice versa.

As we mentioned in the previous part, one of our main issue is to take into account the traditional RSS result and the balance of the stack chain. This explains why hereafter we choose to display the coefficient $f$ with respect to some balance indicator such as $S_1$ or other dispersal measures within input features.

First, in Figure~\ref{fig:link}, we propose to show the results for the coefficient $f$ obtained from a Monte Carlo simulation of uniform distributions with $2 \times 10^5$ drawn observations. Next, we display the coefficient $f$ according to the Chernov methodology as detailed in Proposition~\ref{prop2}. Finally, we show that boundings by Lipschitz and Quadratic approaches directly depend on some balance factor.

\begin{figure}[H]
\centering
\includegraphics[scale=0.62]{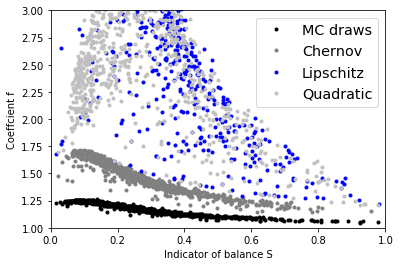}
\caption{Link between the coefficient $f$ and the balance factor $S_1$ with parameter $\rho=0.05$.}
\label{fig:link}
\end{figure} 

We observe an almost linear behavior of the result with respect to the balance factor $S_1$ for the Monte Carlo approach and for the Chernov methodology. As expected and due to the upper bounds defined in these methods, both Lipschitz and Quadratic approaches give results much more conservative. Still we can observe that Quadratic is more accurate for small values of $S_1$. This is explained by the fact that the result with Quadratic approach in Proposition \ref{prop_var} takes into account the variance of input feature bounds. For a small $S_1$, input features are balanced and variance is a more regular control quantity for the structure of the stack chain than the sum of absolute deviations around the mean $\bar{v}$ introduced in Proposition~\ref{prop_lipschitz}.

\subsection{Case study}

In this part, we focus on industrial practices at Airbus. First, we take the example of an assembly from an aircraft and show results from the methodology proposed in this article. Then,  we  detail the common process of tolerance definition at Airbus and we explain how it is related to the approaches presented in the paper. Finally, we represent all stack chains in a real aeronautical product perimeter according to the balance factor.

\subsubsection{An assembly example}

The assembly in Figure~\ref{Ex frame misalignment} is related to a genreic frame misalignment for an Airbus 
aircraft. 

Table~\ref{stack chain} gives the stack chain data of this requirement. Tolerance bounds value have been modified.
It involves 10 input features in the assembly and tolerance data are scaled and unit free. Table~\ref{Res_wc_rss} provides traditional tolerancing worst case and RSS results. The application of the different methods we proposed in this paper gives the results depicted in Figure~\ref{Real case}.

\begin{figure}[H]
\centering
\includegraphics[scale=0.4]{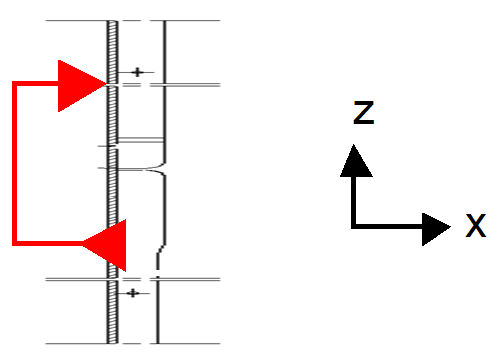}
\caption{Example of vertical frame misalignment with respect to the last rigid point.}
\label{Ex frame misalignment}
\end{figure}

\begin{table}[H]
\caption{Stack chain of the top level requirement: frame misalignment - last rigid point.}\label{stack chain}
\centering
\begin{tabular}{l|c}
%\hline
\multicolumn{1}{c|}{\textbf{Name of the contributor}}     & \textbf{\begin{tabular}[c]{@{}c@{}}Tolerance\\ interval\end{tabular}} \\ \hline
    Frame 1                               & $\pm 1$                                                                           \\ \hline
    Frame 2                       & $\pm 0.5$                                                                             \\ \hline
    Process tolerance                   & $\pm 0.25$                                                                          \\ \hline
    Process tolerance                                & $\pm 0.23$                                                                          \\ \hline
    Process tolerance                            & $\pm 0.2$                                                                          \\ \hline
    Process tolerance                            & $\pm 0.2$                                                                           \\ \hline
    Process tolerance             & $\pm 0.15$                                                                           \\ \hline
    Process tolerance               & $\pm 0.13$                                                                          \\ \hline
    Process tolerance    & $\pm 0.1$                                                                          \\ \hline
    Process tolerance                           & $\pm 0.09$                                                                          \\% \hline
\end{tabular}
\end{table}

\begin{table}[H]
\caption{Result for the top level requirement: frame misalignement - last rigid point.}
\label{Res_wc_rss}
\centering
\begin{tabular}{c|c}
%\hline
Worst Case result & RSS result    \\ \hline
$\pm 2.85$       & $\pm 1.23$ \\ %\hline
\end{tabular}

\end{table}

\begin{figure}[H]
\centering
\includegraphics[scale=0.63]{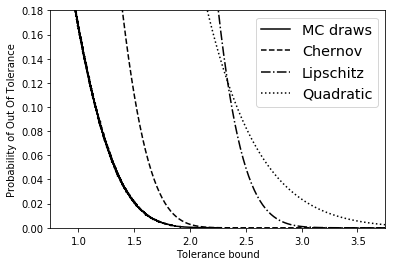}
\caption{Real case study of a bound value according to the confidence level.}
\label{Real case}
\end{figure}

\subsubsection{Industrial practices : Airbus}

One of the methods used by Airbus to define a tolerance in the design phase is based on Monte Carlo simulation data and a disproportion parameter. As for the Gaussian case, quantile at $0.27\%$ are observed on Monte Carlo simulations and a linear regression with respect to the factor parameter is carried out to obtain the result. For a set of tolerance bound $v_1, \dots, v_n > 0$ for an input features balance ratio $D$, this rule gives an output feature tolerance interval $[-T_{Airbus}, T_{Airbus}]$ defined as : 
\begin{equation}
T_{Airbus} = 1.6  \times (-0.56  D + 1.04 ) \times T_{RSS}
\label{Airbus_rule}
\end{equation}
with $T_{RSS}$ as defined in previous parts and 
\begin{equation*}
\forall  v_1, \dots, v_n > 0, \quad
D = \frac{\max_{i}(v_i) - \bar{v}}{\sum_{i=1}^n v_i} .
\end{equation*}

This $D$ factor measures how far from the mean is the main contributor of the stack chain and has the advantage of being understandable. This quantity is highly correlated to the term $S_1$ previously introduced as we can see on the Figure~\ref{cor} :

\begin{figure}[H]
\centering
\includegraphics[scale=0.63]{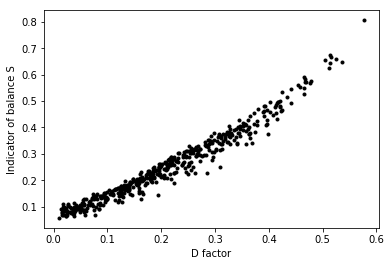}
\caption{Correlation between the term $S_1$ and the balance factor $D$.}
\label{cor}
\end{figure} 

With this definition, a high $D$ still implies an unbalanced stack chain. Conversely, a small value of $B$ means a balanced between stack chain inputs. Figure~\ref{D expl} shows a few examples of this factor with the respect to the stack chain structure.

\begin{figure}[H]
\centering
\includegraphics[scale=0.63]{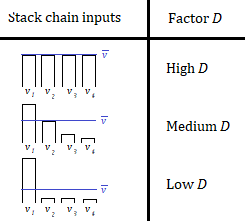}
\caption{Stack chain structure and balance factor $D$.}
\label{D expl}
\end{figure} 

In the industrial example of the frame misalignment. We set the probability to be out of tolerance for the output feature at $\rho = 0.0027$, which corresponds to the acceptable $0.27\%$ out of the interval from the $6\sigma$ methodology. Table~\ref{tol result} summarizes the tolerance interval obtained for the frame misalignment. Three results are displayed : the Monte Carlo approach with $200000$ drawn observations for each input feature, the Chernov approach proposed in this article with $\rho=0.0027$ and the industrial practice presented in \eqref{Airbus_rule}.

\begin{table}[H]
\caption{Tolerance interval results according to the different approaches.}\label{tol result}
\centering
\begin{tabular}{c|c|c|c}
\textbf{Method}  & \begin{tabular}[c]{@{}c@{}}Monte Carlo\\ $\rho=0.27\%$\end{tabular} & \begin{tabular}[c]{@{}c@{}}Chernov\\ $\rho=0.27\%$\end{tabular} & \begin{tabular}[c]{@{}c@{}}Industrial practice\\ no $\rho$ selection\end{tabular} \\ \hline
\textbf{\begin{tabular}[c]{@{}c@{}}Tolerance\\ interval\end{tabular}} & $\pm3.56$mm & $\pm4.01$mm    & $\pm3.53$mm                                                           \\ 
\end{tabular}
\end{table}

We can see that for a level $\rho=0.27\%$, the result from the industrial rule is very close to the value observed on Monte Carlo simulations. The Chernov method gives a more conservative result but ensures a precise probability $\rho$ to be out of the interval for the output feature.

\subsubsection{Performance of the different approaches on industrial cases}
Focusing on a real sample of aeronautical assemblies, all stack chains have been analyzed in order to obtain the $f$ coefficient according to the different methodologies.

\begin{figure}[H]
\centering
\includegraphics[scale=0.63]{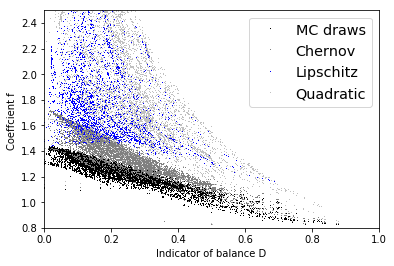}
\caption{Link between the term introduced and the balance factor $D$.}
\label{Link between the term introduced and the balance factor $D$}
\end{figure} 

We retrieve the same trend that for simulated data : Better result for Chernov approach and linearity in $D$.

\section{Conclusion}
We proposed robust approaches for tolerance definition in the design phase allowing the management of confidence level. From known input tolerance intervals of an assembly and for a selected confidence level, we can determine the output tolerance interval. The result will be robust against poor or unknown industrial capabilities because uniform distributions on tolerance intervals are assumed for input features.\\

The Chernov method is particularly accurate and gives an output tolerance interval result close to the reality, tight enough to be industrially relevant, and ensures also the selected probability as confidence level. We also provided a balance factor which is strongly related to how tight an interval should be according to the disproportion of the stack chain. We obtained an almost linear behavior of our result from the Chernov methodology with respect to this balance factor.\\

Future directions of this work would be to considere more adversarial distributions for input features. For instance, bimodal distributions or truncated distributions could be studied in order to hedge against industrial practices with the induced bias of machine or thrust effect.

\section*{Proof of Lipschitz continuity of function $h$}

We start by recalling the definition of Lipschitz continuity of a function $f: \mathbb{R} \to \mathbb{R}$. Let $L > 0$, if the function is such that
\begin{equation*}
\forall x, y \in \mathbb{R},\ \left| f(x) - f(y) \right| \leqslant L \left| x - y\right|
\end{equation*}
then, $f$ is said to be Lipschitz continuous with constant $L$. In this appendix, we propose to prove the previously claimed Lipschitz continuity with constant $L=1/2$ of function $h$ defined by
\begin{equation*}
\forall x >0, \quad  h(x)= \log \left( \frac{1-e^{-x}}{x} \right).
\end{equation*}

The first and second derivative functions of $h$ are easily obtained as
\begin{equation*}
\forall x >0, \quad  h'(x)= \frac{e^{-x}}{1 - e^{-x}} - \frac{1}{x}
\end{equation*}
and
\begin{equation*}
\forall x >0, \quad  h''(x)=  \frac{1 + (1+x^2)e^{-x}}{x^2(1-e^{-x})^2}.
\end{equation*}
Since $h''(x) \geqslant 0$ for any $x > 0$, then $h'$ is an non decreasing function. Moreover, $h'(x)$ tends to $-1/2$ when $x \to 0^+$ and to $0$ when $x \to +\infty$. Thus, we conclude that $\left|  h'(x) \right| \leqslant \frac{1}{2}$ and finally that $h$ is Lipschitz continuous with $L=1/2$ by a straightforward integration argument.

\section*{Acknowledgment}
Industrial application detailed in the article is the topic of a patent application FR 15508 (patent pending n\textsuperscript{o}1912668).

\vspace{\baselineskip}

\noindent This work was partly supported by the French \textit{Agence Nationale de la Recherche} through a CIFRE contract 2017/1390 in partnership with Airbus.


\begin{thebibliography}{1}

\bibitem{ASME Y14}
ASME Y14.5M-1994, \emph{Dimensioning and Tolerancing}, The American Society of Mechanical Engineers, 1994.

\bibitem{ISO 1101}
ISO 1101:2017, \emph{Geometrical product specifications (GPS) — Geometrical tolerancing — Tolerances of form, orientation, location and run-out}, ISO/TC 213 - International Organization for Standardization, 2017.


\bibitem{CHASE 88}
K.~W. Chase and W.~H. Greenwood,  \emph{Design issues in mechanical tolerance analysis}, Manufacturing Review, 1988, 1(1), pp. 50-59.

\bibitem{BENDER 62}
A.~Bender, \emph{Benderizing tolerances-a simple practical probability method of handling tolerances for limit-stack-ups}, Graphic Science, 1962, vol. 17.

\bibitem{SKO TURNER 97}
V.~J. Skowronski and J.~U. Turner,  \emph{Using Monte-Carlo variance reduction in statistical tolerance synthesis}, Computer-Aided Design, 1997, vol. 29, no 1, p. 63-69.

\bibitem{CHOI 00}
H.-G.~R.Choi, M.-H.~Park and E.~Salisbury,  \emph{Optimal tolerance allocation with loss functions}, Journal of Manufacturing Science and Engineering, 2000, vol. 122, no 3, p. 529-535.

\bibitem{PILLET 05}
M.~Pillet, D.~Duret and A.~Sergent,  \emph{Weighted inertial tolerancing}, Quality Engineering, 2005, vol. 17, no 4, p. 687-693.

\bibitem{LEBLOND PILLET 18}
L.~Leblond and M.~Pillet,  \emph{Conformity and statistical tolerancing},   International Journal of Metrology and Quality Engineering 9, 2018: 1.

\bibitem{DRAKE 99}
P.~J. Drake, P.~J. Drake and V.~Srinivasan,  \emph{Dimensioning and tolerancing handbook}, New York : McGraw-Hill, 1999, ch. 8.

\bibitem{KILLMANN 01}
F.~Killmann and E.~Von  Collani,  \emph{A note on the convolution of the uniform and related distributions and their use in quality control},  Economic Quality Control, 2001, vol. 16, no 1, p. 17-41.

\bibitem{VILLANI 03}
C.~Villani, \emph{Topics in optimal transportation}, Issue 58 of Graduate studies in mathematics, ISSN 1065-7339, American Mathematical Soc., 2003.

\bibitem{HOEFFDING 63}
W.~Hoeffding,  \emph{Probability Inequalities for Sums of Bounded Random Variables}, Journal of the American Statistical Association, 58:301, 13-30, DOI: 10.1080/01621459.1963.10500830, 1963.

\bibitem{BOUCHERON 13}
S.~Boucheron, G.~Lugosi and P.~Massart,  \emph{Concentration inequalities: A nonasymptotic theory of independence}, Oxford university press, 2013.

\bibitem{NESTEROV 13}
Y.~Nesterov, Yurii, \emph{Introductory lectures on convex optimization: A basic course.} Vol. 87. Springer Science \&  Business Media, 2013.
\end{thebibliography}
\end{document}